\documentclass[leqno]{amsart}
\usepackage{amssymb}
\usepackage{color}
\usepackage{mathrsfs}
\usepackage{amsmath, amsfonts, vmargin, enumerate}
\usepackage{graphics}
\usepackage{verbatim}
\usepackage{amsthm}
\usepackage{latexsym,bm}
\usepackage{euscript}

\input xypic
\xyoption {all}

 \makeatletter

\def\>{\rangle}
\def\<{\langle}

\newtheorem{thm}{Theorem}[section]
\newtheorem{prop}[thm]{Proposition}

\newtheorem{lem}[thm]{Lemma}
\newtheorem{fact}[thm]{Fact}

\theoremstyle{definition}
\newtheorem{defi}[thm]{Definition}

\theoremstyle{remark}
\newtheorem{remark}[thm]{Remark}
\newtheorem{example}[thm]{Example}
\newtheorem{pb}[thm]{Problem}

\newenvironment{definition}{\begin{defi}\rm}{\end{defi}}

\numberwithin{equation}{section}

\newcommand{\un}{1\mkern -4mu{\rm l}}

\renewcommand{\l}{\lambda}

\newcommand{\Tr}{\mbox{\rm Tr}}

\newcommand{\8}{\infty}

\newcommand{\be}{\begin{eqnarray*}}
\newcommand{\ee}{\end{eqnarray*}}
\newcommand{\beq}{\begin{equation}}
\newcommand{\eeq}{\end{equation}}
\newcommand{\beqn}{\begin{equation*}}
\newcommand{\eeqn}{\end{equation*}}

\begin{document}

\title{Operator space approach to steering inequality}


\author{Zhi Yin}
\address{School of Mathematics and Statistics, Wuhan University, Wuhan 430072, China; Institute of Theoretical Physics and Astrophysics, Gda{\'n}sk University, Wita Stwosza 57, 80-952 Gdansk, Poland}
\email{hustyinzhi@163.com}

\author{Marcin Marciniak}
\address{Institute of Theoretical Physics and Astrophysics, Gda{\'n}sk University, Wita Stwosza 57, 80-952 Gda{\'n}sk, Poland}
\email{matmm@ug.edu.pl}

\author{Micha{\l} Horodecki}
\address{Institute of Theoretical Physics and Astrophysics, Gda{\'n}sk University, Wita Stwosza 57, 80-952 Gda{\'n}sk, Poland; National Quantum Information Centre of Gda{\'n}sk, Andersa 27, 81-824 Sopot, Poland}
\email{fizmh@ug.edu.pl}

\date{\today}

\begin{abstract}
In \cite{JP2011,JPPVW2010} the operator space theory was applied to study bipartite Bell inequalities. The aim of the paper is to follow this line of research and use the operator space technique to analyze the steering scenario. We obtain a bipartite steering functional with unbounded largest violation of steering inequality, as well as we can construct all ingredients explicitly. It turns out that the unbounded largest violation is obtained by non maximally entangled state. Moreover, we focus on the bipartite dichotomic case where we %can always find
construct
a steering functional with unbounded largest violation of steering inequality. This phenomenon is different to the Bell scenario where only bounded largest violation can be obtained by any bipartite dichotomic Bell functional.%inequality.
\end{abstract}
\maketitle

\markboth{Z. Yin, M. Marciniak and M. Horodecki}%
{Operator space approach to steering inequality}

%\bigskip

%%%%%%%%%%%%%%%%%%%%%%%%%%
%%%%%%%%%%%%%%%%%%%%%%%%%%

\section{Introduction}

The violation of local realism, called usually nonlocality, plays an important role in quantum information science. It was first studied in 1935 by Einstein, Podolsky and Rosen \cite{EPR1935}. In their paradoxical paper, Einstein, Podolsky, and Rosen (EPR) argued that quantum mechanics does not provide a complete description of the elements of reality. Moreover, they predicted either quantum mechanics develops to a complete theory or quantum mechanics is replaced by another complete theory. In the paper, EPR wrote: ``We believe that such (complete) a theory is possible." Theories compatible with the EPR's ideas are called ``local-realistic (LR) theories". Although Bohr rebutted shortly after the EPR's ideas, their arguments, without observational consequences, did not suggest a clear conclusion, hence the debate has then subsided.

The EPR arguments were resurfaced by J. S. Bell in 1964 when he derived a constraint for correlation between two remote subsystems, known as Bell's inequality which is satisfied by all LR theories. He proved that it is violated by quantum correlations of two spin 1/2 particles in a singlet state, i.e., an entangled state. States that violate some Bell inequalities form strict subset of the set of entangled states \cite{Werner1989}.
%We will call the quantum states have Bell nonlocality when the states can violate some Bell inequalities. These kind of quantum states form a %strictly subset of entangled states \cite{Werner1989}.
In \cite{WJD2007} the authors proposed an intermediate form of quantum correlations between Bell nonlocality and entanglement, by use
of {\it quantum steering}. The latter concept was introduced by Schr\"{o}dinger in 1935 to reply the EPR paradox \cite{Sch1936}. Wiseman, Jones and Doherty reformulated this concept in a rigorous way \cite{WJD2007} and have, in particular,  shown that the set of states admitting steering is a strict subset of entangled states on one hand and a strict superset of states violating Bell inequalities. Since then, quantum steering has attracted more and more attention both in theory \cite{CYWSCKO2013,NG2012,Pusey2013,WJD2007} and experiment \cite{SJWP2010,Smith2012}.

The simplest example of quantum steering is the following one, which was a basis for famous EPR paradox \cite{EPR1935}
(in Bohm version \cite{Bohm1951}). Namely, when  Alice and Bob share a pair of particles in singlet state, Alice,
by choosing one of two measurements can create at Bob's site one of two ensembles:
one consisting of basis states $|0\>$ and $|1\>$ with equal probabilities and the other, consisting of
complementary states $\frac{1}{\sqrt2}(|0\>+|1\>)$ and $\frac{1}{\sqrt2}(|0\>-|1\>)$, again with equal probabilities.
It turns out that this would be impossible, if Bob particle were in some well defined state, perhaps unknown to him -- so called ``local hidden state" (LHS),
and Alice merely used her knowledge about the state. Thus, existence of the above Alice's measurements proves that the shared state is entangled. Remarkably, Alice can in this way convince Bob, that the shared state is entangled even if Bob does not trust her.
Indeed, Bob can ask Alice to create one of the two ensembles at random, and
upon receiving message from Alice telling which outcome she obtained,
he can verify that indeed she created (or: ``steered" to) the above states with the mentioned probabilities,
provided many runs of the experiment are performed.
More generally, bipartite states for which there exist measurements of Alice steering to ensembles,
which doesn't come from local hidden state  are called steerable (or admitting quantum steering).

In steering scenario, one can study ``steering functionals" which are analogs of Bell functionals. The violation of a steering inequality for such functionals provides a natural way to quantify the deviation from a LHS description. However it is not easy to compute the violation for a given steering functional analytically; one usually uses here numerical method, called semi-definite program \cite{Pusey2013,SNC2013}. In this paper, we will use the {\it operator space theory}, which has been widely developed after the pioneering and fundamental work of Effros-Ruan and Blecher-Paulsen \cite{ER2000,Pisier2003}, to study the violation of steering inequality. Our work is motivated by a series works of Junge, Palazuelos, P\'{e}rez-Garc\'{i}a, Villanueva and Wolf \cite{JP2011,JPPVW2010,PWPVM2008}, where they used operator space theory to analyze Bell inequalities. We will briefly recall their work in subsection 2.2. In their work, operator space was connected to the largest violation of Bell inequality. As steering inequality is closely connected to Bell inequality, we are able to apply their strategy. %to study steering functional.
According to their work, following results are natural:
We can construct a probabilistic steering functional with unbounded largest violation of steering inequality in the sense of ``with high probability". Non maximally entangled state plays an important role in violation of steering inequality for this functional.

Moreover, by mixing white noise with this non maximally entangled state, we can get a family of PPT (Positive Partial Transpose) states. Actually, these states belong to a class of PPT states, which was constructed in \cite{CK2006}. We prove that for any steering functional, only bounded largest violation of steering inequality can be obtained by these kinds of PPT states. Thus even though there are PPT states violating some steering inequalities \cite{MGHG2014} (problem posed in \cite{Pusey2013}), our example provides some evidence, that PPT states cannot provide unbounded violation.

However, not all properties of Bell functionals can be inherited by steering ones. It is well known that the quantum bound of any bipartite dichotomic Bell functional is bounded by the classical bound with an universal constant (Grothendieck constant) \cite{Pisier2012,PWPVM2008}. As reported in a companion paper \cite{steering_short} in steering scenario, this is not true: there is a bipartite dichotomic steering functional with unbounded largest violation. Here we put the example of \cite{steering_short} into the framework of operator spaces and show how the inequality arises from a bounded but not completely bounded
map from $\ell_\8^n$ to $\mathbb{M}_n.$

In this paper, operator spaces and their tensor products are the main mathematical tools. For more information about operator space theory, see \cite{ER2000,Pisier2003}. We will use various kinds of operator spaces, such as $\ell_{\8}^n, \ell_1^n, \ell_1^n(\ell_\8^n), R_n, C_n$ and $OH_n$. We strongly recommend readers to refer \cite{JP2011,JPPVW2010} for these notions.

We finish this introduction by setting the following convention: throughout this paper, we will use $\gtrsim$ and $\lesssim$ to denote the inequality up to an universal constant irrelevant to $n\in \mathbb{N},$ and we also use the Dirac symbol $|i\rangle\langle j|, i,j=1,\ldots,n$ to denote the canonical basis of $\mathbb{M}_n.$ For $k\in\mathbb{N}$ we denote by $(e_1,\ldots,e_k)$ the canonical basis of $\ell_1^k$ and by $(e_1',\ldots,e_k')$ its dual basis in $\ell_\8^k$. Then the system $(e_x\otimes e_a')_{x=1,\ldots,n;\,a=1,\ldots,m}$ forms a basis of the space $\ell_1^n(\ell_\infty^m)$ as well as the system $(e_x'\otimes e_a)_{x=1,\ldots,n\;\,a=1,\ldots,m}$ is a basis of $\ell_\8^n(\ell_1^m)$.

\vspace{5mm}

\section{Main result}
\subsection{Definition of steering functional and its largest violation}
We consider the following steering scenario \cite{Pusey2013}. Assume that there are two systems  A (Alice) and B (Bob). Suppose Alice can choose among $n$ different measurement settings labeled by $x=1,\ldots,n$. Each of them can result in one of $m$ outcomes, labeled by $a=1,\ldots, m.$ Suppose also that Bob has a $d$-dimensional quantum system $H_d$ at his disposal.
\begin{definition}
\label{d:assem}
An {\em assemblage} is a set $\{\sigma_x^a:\,x=1,\ldots,n,\;a=1,\ldots,m\}$ of $d \times d$ Hermitian matrices  satisfying the following conditions:
\begin{enumerate}%[{\rm i)}]
\item[i)]
$\sigma_x^a \geq 0$ (positivity);
\item[ii)]
$\sum_a \sigma_x^a$ is independent of $x$ and trace 1.
\end{enumerate}

We say that a set $\{\sigma_x^a\}$ of positive matrices forms an {\em incomplete assemblage} if it satisfies the following condition instead of the above condition ii)
\begin{enumerate}
\item[ii')]
$\Tr\left(\sum_a\sigma_x^a\right)\leq 1$ for every $x$.
\end{enumerate}
\end{definition}

It turns out (\cite{Sch1936} and later \cite{HJW1993}) that any assemblage (respectively incomplete assemblage) has a quantum realization.
It means that for any assemblage
%We say that an assemblage is a \textit{quantum assemblage} (respectively \textit{incomplete quantum assemblage}), if
there exist a Hilbert space $H$, a density matrix $\rho\in B(H\otimes H_d)$ and a family $\{E_x^a\}_a$ of positive operators on $H$ such that $\sum_a E_x^a =\un$ (respectively $\sum_a E_x^a\leq \un$) for every $x$, and
        \beq \sigma_x^a= \Tr_A ((E_x^a \otimes \un_B) \rho).\eeq
We denote the set of all quantum assemblages (respectively incomplete quantum assemblages) by $\mathcal{Q}$ (respectively by $\mathcal{Q}^{\mathrm{in}}$).

Next, we distinguish the classical part of the set assemblages.
\begin{definition}
We say that an assemblage (respectively incomplete assemblage) admits a \textit{local hidden state (LHS)} model, if
there exist finite set of indices $\Lambda$, nonnegative coefficients $q_\lambda$ such that $\sum_\lambda q_\lambda=1$, density matrices $\sigma_\lambda\in B(H_d)$ for $\lambda\in\Lambda$, and nonnegative numbers $p_\lambda(a|x)$ such that $\sum_ap_\lambda(a|x)=1$ (respectively $\sum_a p_\lambda(a|x)\leq 1$) for every $x,\lambda$, and
\beq
\sigma_x^a = \sum_{\lambda\in\Lambda} q_\lambda p_\lambda (a|x) \sigma_\lambda
\eeq
for every $x$ and $a$.
We denote the set of LHS assemblages (respectively LHS incomplete assemblages) by $\mathcal{L}$ (respectively $\mathcal{L}^{\mathrm{in}}$).
\end{definition}

It is known \cite{WJD2007} that $\mathcal{L} \subsetneq \mathcal{Q}$.
Our aim is to quantify the difference between sets $\mathcal{Q}$ and $\mathcal{L},$ and our strategy is to analyze some functionals to see how much value they can obtain on the quantum assemblages comparing with values on LHS assemblages.
\begin{definition}
For given natural numbers $n, m$ and $d$, define (\cite{Pusey2013}) a \textit{steering functional} $F$ as a set $\{F_x^a:\;x=1,\ldots,n,\;a=1,\ldots,m\}$ of $d \times d$ matrices.
The functional maps an assemblage $\sigma$ to a real number
\begin{equation}
\label{e:action}
\langle F,\sigma\rangle =\sum_{x=1}^n\sum_{a=1}^m \Tr(F_x^a \sigma_x^a).
\end{equation}
%Given an assemblage $\sigma\in\calQ$ we now define the \textit{largest steering violation} that  $\sigma$ may attain as:
%\beq \nu(\sigma)=\sup \left\{ \left|\langle F,\sigma\rangle\right|:\mbox{$F$ is a steering functional such that $|\langle F,\tau\rangle|\leq 1$ for $\tau\in\calL$.}\right\}.\eeq
\end{definition}

%\begin{defi}
%\label{d:incomplete}

%We say that $\sigma=\{\sigma_x^a\}$ is an \textit{incomplete quantum assemblage}, if there exists a Hilbert space $H$ such that
%\beq
%\sigma_x^a= \Tr_A ((E_x^a \otimes \un_B) \rho)
%\eeq
%for every $x$, $a$, where $\rho \in B(H \otimes H_d)$ is a density matrix and $\{E_x^a\}_{a}\subset B(H)$ is an incomplete POVM measurement on Alice for any $x$, i.e. $E_x^a \geq 0$ for every $x, a$ and $\sum_a E_x^a \leq \un$ for every $x$.

%On the other hand, we say that an assemblage $\sigma$ has an \textit{incomplete local hidden state (LHS)} model, if
%        \beq \sigma_x^a = \sum_{\lambda\in\Lambda} q_{\lambda} p_\lambda(a|x) \sigma_\lambda\eeq
%for every $x, a$, where $\Lambda$ is a finite set of indices, $q_\lambda$ are nonnegative numbers such that $\sum_\lambda q_\lambda=1$, $\sigma_\lambda\in B(H_d)$, $\sigma_\lambda\geq 0$, $\Tr(\sigma_\lambda)\leq 1$ for all $\lambda$, and $p_\lambda(a|x)$ are nonnegative numbers for all $a, x, \lambda$, such that $\sum_a p_\lambda(a|x)\leq 1$ for all $x, \lambda$.

%We denote the set of incomplete quantum assemblage by $\mathcal{Q}^\mathrm{in}$ and the set of incomplete LHS by $\mathcal{L}^\mathrm{in}.$
%\end{defi}
\begin{remark}
The notion of steering inequality was first introduced by E. G. Cavalcanti, S. J. Jones, H. M. Wiseman, and M. D. Reid \cite{CJWR2009}. Given a steering functional, a steering inequality says that
\begin{equation}
\langle F,\sigma\rangle \leq B_C(F),\qquad \sigma\in\mathcal{L}^{\mathrm{in}}.
\end{equation}
\end{remark}

Now we can define
\begin{defi}
Given a steering functional $F=\{F_x^a\in B(H_d): x=1,\ldots,n,\; a=1,\ldots,m\},$ we define the \textit{LHS bound} of $F$ as a number
\beq B_C (F)=\sup \{|\langle F,\sigma\rangle|: \sigma \in \mathcal{L}^\mathrm{in} \},\eeq
and the \textit{quantum bound} of $F$ as
\beq B_Q (F)=\sup \{ |\langle F,\sigma\rangle|: \sigma\in \mathcal{Q}^\mathrm{in} \}.\eeq
We define the \textit{largest quantum violation of steering inequality} for $F$ as a positive number
\beq LV(F)=\frac{B_Q(F)}{B_C(F)}.\eeq
In the sequel we will frequently call this number shortly \textit{largest violation} for $F$.
\end{defi}

\subsection{Junge-Palazuelos approach to violation of Bell inequality}

M. Junge, C. Palazuelos, D. P\'{e}rez-Garc\'{i}a, I. Villanueva and M. M. Wolf studied the following largest violation of Bell inequality for a bipartite Bell functional $M$ by using operator space theory \cite{JP2011,JPPVW2010}:
\beq
LV(M)= \frac{B_Q(M)}{B_C(M)},
\eeq
where $M=\sum_{x,y=1}^{n}\sum_{ a,b=1}^{m} M_{x,y}^{a,b}\, e_x \otimes e_a' \otimes e_y \otimes e_b' \in \ell_1^n(\ell_\8^m)\otimes \ell_1^n(\ell_\8^m),$ and
\beq
B_Q(M)= \sup \Big\{ \Big| \sum_{x,y,a,b} M_{x,y}^{a,b} \Tr((E_x^a \otimes E_y^b ) \rho ) \Big|: \mbox{$E_x^a, E_y^b$ are incomplete POVMs, $\rho$ is a state},\Big\}
\eeq
\beq
B_C(M)= \sup \Big\{ \Big| \sum_{x,y,a,b} M_{x,y}^{a,b} \sum_{\lambda} \rho(\lambda)\mathbb{P}(a|x,\lambda)\mathbb{P}(b|y,\lambda) \Big|: \sum_{a}\mathbb{P}(a|x,\lambda) \leq 1, \sum_b \mathbb{P}(b|y,\lambda)\leq 1 ,\Big\}
\eeq

They linked the largest violation of Bell inequality to the {``\it min vs $\epsilon$"} problem of $M,$ i.e, they obtained following result (See \cite[Propostion 4, Theorem 6]{JPPVW2010}):

\begin{prop}\label{prop:transfer}
Given $M= \sum_{x,y=1}^{n}\sum_{ a,b=1}^{m} M_{x,y}^{a,b}\, e_x \otimes e_a' \otimes e_y \otimes e_b' \in \ell_1^n(\ell_\8^m)\otimes \ell_1^n(\ell_\8^m),$ we have the following equivalence:
\begin{enumerate}[{\rm i)}]
\item Classical bound:
\begin{equation}\label{eq:classical}
B_C (M) \approx \|M\|_{\ell_1^n(\ell_\8^m)\otimes_{\epsilon} \ell_1^n(\ell_\8^m)}.\end{equation}
\item Quantum bound:
\begin{equation}\label{eq:quantum}
B_Q (M) \approx \|M\|_{\ell_1^n(\ell_\8^m)\otimes_{\min} \ell_1^n(\ell_\8^m)}.\end{equation}
\end{enumerate}
\end{prop}

Here the notion $X\otimes_\epsilon Y$ denotes the injective tensor product of two Banach spaces $X$ and $Y$ \cite{Tak1979}. And for two operator spaces $E$ and $F$, $E\otimes_{\min} F$ denotes the minimal tensor product of operator space \cite{Pisier2003}. $\ell_1^n(\ell_\8^m)$ is a Banach space in \eqref{eq:classical} and it is an operator space in \eqref{eq:quantum}. However, it is difficult to calculate the tensor norm of element in $\ell_1^n(\ell_\8^m)\otimes \ell_1^n(\ell_\8^m)$. To overcome this obstacle, they used following fact (see \cite[Theorem 3.2]{JP2011}):
\begin{fact}\label{fact:map}
There exist $\delta \in (0,\frac{1}{2})$ and an universal constant $C$ such that, for every $n$, there are maps $V: \ell_2^{[\delta n]} \to \ell_1^n(\ell_\8^n)$ and $V': \ell_1^n(\ell_\8^n) \to \ell_2^{[\delta n]}$ satisfying $\|V\|\leq C \sqrt{\log n}$, $\|V'\|\leq 1$ and $V' V= id_{\ell_2^{[\delta n]}}$. Moreover, the map $V'$ is completely bounded from $\ell_1^n(\ell_\8^n)$ to $R_{[\delta n]}$ and $\|V'\|_{cb}\leq K_L \|V'\| \lesssim 1$, where $K_L$ is the constant in little Grothendieck theorem.
\end{fact}

Compare to $\ell_1^n(\ell_\8^n),$ the tensor norms on $\ell_2^n$ are easier to calculate. Through this approach, they obtained a Bell inequality with unbounded largest violation of order $\frac{\sqrt{n}}{\log n}.$

\subsection{Unbounded largest violation of steering inequality}

It is well known that there is a very close relation between Bell inequality and steering one \cite{CJWR2009}. So it is not surprising for us to find a steering functional with unbounded largest violation of steering inequality through their approach. Our main result concerns the case when all the numbers $n$, $m$ and $d$ are equal. It can be stated as follows:

\begin{thm}\label{thm:LargeViolation}
For every $n\in \mathbb{N},$ we can find a steering functional $F=\{F_x^a\in B(H_n): x,a=1, \ldots, n\},$ such that
\beq LV (F)\gtrsim \frac{\sqrt{n}}{\sqrt{\log{n}}}.\eeq
\end{thm}

Let $F=\{F_x^a\}_{x,a}$ be a steering functional. We will identify $F$ with the following element of the tensor product $\ell_1^n(\ell_\infty^m)\otimes B(H_d)$:
\beq \sum_{x=1}^n\sum_{a=1}^m (e_x \otimes e_a') \otimes F_x^a.\eeq
We still use $F$ to denote this element.

Now, in the spirit of Proposition \ref{prop:transfer}, we can link the problem of largest violation of steering functional to the {``min vs $\epsilon$"} problem for $F$.

\begin{prop}\label{prop:minandepsilon}
Given $F= \sum_{x,a} (e_x \otimes e_a') \otimes F_x^a \in \ell_1^n(\ell_\8^m) \otimes B(H_d),$ we have the following equivalence:
\begin{enumerate}[{\rm i)}]
\item LHS bound:
\begin{equation}\label{eq:classical}
B_C(F)\leq \|F\|_{\ell_1^n(\ell_\8^m) \otimes_{\epsilon} B(H_d)} \leq 16 B_C(F).\end{equation}
\item Quantum bound:
\begin{equation}\label{eq:quantum}B_Q (F) \leq \|F\|_{\ell_1^n(\ell_\8^m) \otimes_{\min} B(H_d)} \leq 4 B_Q(F).\end{equation}
\end{enumerate}
\end{prop}

\begin{proof}
For the LHS bound, we will use the duality between injective and projective tensor product for finite dimensional Banach spaces. For a Hilbert space $H,$ we denote by $S_p(H)$ the $p$-th Schatten class for $H$, where $p$ is a number such that $p\geq 1$. Any element $\sigma\in \ell_\8^n(\ell_1^m)\otimes_\pi S_1(H_d)=\left(\ell_1^n(\ell_\8^m)\right)^*\otimes_\pi\left(B(H_d)\right)^*$
can be considered as a functional on $\ell_1^n(\ell_\8^m) \otimes_{\epsilon} B(H_d)$.
Its action on $F$ is given by
\beq
\langle F,\sigma\rangle=\sum_{x=1}^n\sum_{a=1}^m\Tr(F_x^a\sigma_x^a),
\eeq
where matrices $\sigma_x^a\in B(H_d)$ are determined by the unique decomposition $\sigma=\sum_{x,a}e_x'\otimes e_a\otimes \sigma_x^a$.
Observe that the action of a steering functional on an assemblage given by (\ref{e:action}) is a special case of above duality.
Given a Banach space $X,$ let $\mathbb{B}_X$ denote the unit ball of $X$.
Thus, by the duality, we have
\beq
\|F\|_{\ell_1^n(\ell_\8^m)\otimes_{\epsilon} B(H_d)}=\sup \left\{|\langle F,\sigma\rangle|: \; \sigma \in \mathbb{B}_{\ell_\8^n(\ell_1^m)\otimes_{\pi} S_1(H_d)} \right\}.
\eeq
Now, observe that
\begin{equation}
\label{e:inclusion}
\mathcal{L}^\mathrm{in}\subset\mathbb{B}_{\ell_\8^n(\ell_1^m)\otimes_{\pi} S_1(H_d)}.
\end{equation}
The first inequality in (\ref{eq:classical}) follows from inclusion (\ref{e:inclusion}).

For the quantum bound, we first recall the following fact from \cite{JPPVW2010}:
\begin{fact}
Given a set of incomplete POVMs  $\{E_x^a\}_{a=1,\ldots,m}$, $x=1,\ldots,n$, on $B(H)$, the operator $u:\ell_1^n(\ell_\8^m) \to B(H)$ defined by $u(e_x \otimes e_a') =E_x^a,$ is a completely contraction.
\end{fact}
The minimal tensor norm of $F$ can be expressed as follows (\cite{Pisier2003}):
\begin{equation}
\label{e:min}
\|F\|_{\ell_1^n(\ell_\8^m) \otimes_{\min} B(H_d)}=\sup_{H,u}  \|(u\otimes \mathrm{id})(F) \|_{B(H)\otimes_{\min}B(H_d)},
\end{equation}
where the $\sup$ is taken over all possible Hilbert spaces $H$ and completely contractions $u: \ell_1^n(\ell_\8^m) \to B(H)$.
Let $\sigma\in\mathcal{Q}^\mathrm{in}$, i.e. there is a Hilbert space $H$, a density matrix $\rho\in B(H)\otimes B(H_d),$ and incomplete POVMs $\{E_x^a\}$ on $B(H)$ such that $\sigma_x^a=\Tr_\mathrm{A}((E_x^a\otimes\un)\rho)$ for every $x,a$. Let $u$ be a map $u: \ell_1^n(\ell_\8^m) \to B(H)$ defined by $u(e_x\otimes e_a)= E_x^a$. By the aforementioned fact $u$ is a completely contraction. Thus, by following claim it is enough to show the first inequality in \eqref{eq:quantum}
\begin{equation}
\label{eq:1}
\langle F,\sigma\rangle = \Tr((u\otimes \mathrm{id})(F) \rho).
\end{equation}
 The second inequality in \eqref{eq:classical} and \eqref{eq:quantum} can be proved by using the same argument in the proof of Proposition \ref{prop:transfer}. Here we omit the details.
\end{proof}

The result which we will prove is
\begin{thm}\label{thm:largeviolation}
For every $n\in \mathbb{N},$ there exists an element $F =\sum_{x,a=1}^{n} (e_x \otimes e_a') \otimes F_x^a\in \ell_1^n(\ell_\8^n) \otimes B(H_n)$ such that
\beq \frac{\|F\|_{\min}}{\|F\|_{\epsilon}} \gtrsim \frac{\sqrt{n}}{\sqrt{\log n}}.\eeq
\end{thm}

Theorem \ref{thm:LargeViolation} follows from this theorem and proposition \ref{prop:minandepsilon}.\\

\begin{proof}
Define a map $W: \ell_2^{[\delta n]} \to B(H_n)$ by:
\begin{equation}
\label{e:W}
W(e_k)= |1\rangle \langle k|, k=1,\ldots, [\delta n].
\end{equation}
It is easy to check $W$ is a contraction, i.e. $\|W\|\leq 1.$

Combining with the fact \ref{fact:map}, we consider the element:
\begin{equation}
\label{e:F}
F=(V\otimes W)(a) \in \ell_1^n(\ell^n_\8) \otimes B(H_n),
\end{equation}
where $a= \sum_{k=1}^{[\delta n]} e_k \otimes e_k.$
On one hand, we have
\begin{equation}\label{eq:epsilon}
\|F\|_{\ell_1^n(\ell_\8^n) \otimes_{\epsilon} B(H_n) }=\|V\otimes W (a)\|_{\ell_1^n(\ell_\8^n) \otimes_{\epsilon} B(H_n)} \leq \|V\|\|W\|\|a\|_{\ell_2^{[\delta n]} \otimes_{\epsilon} \ell_2^{[\delta n]}} \lesssim \sqrt{\log n}.
\end{equation}
On the other hand, the formula \eqref{e:min} implies
\begin{eqnarray}\label{eq:min}
\lefteqn{\|F\|_{\ell_1^n(\ell_\8^n) \otimes_{\min} B(H_n)}  \gtrsim}\\
&\gtrsim& \|(V' \otimes id)(F)\|_{R_{[\delta n]} \otimes_{\min} B(H_n)}= \|(V' \otimes id)(V\otimes W)(a)\|_{R_{[\delta n]} \otimes_{\min} B(H_n)} \nonumber\\
& =&\|\sum_{k=1}^{[\delta n]} e_k \otimes |1\rangle\langle k|\|_{R_{[\delta n]} \otimes_{\min} B(H_n)}= \|\sum_k |1\rangle\langle k| k\rangle\langle 1| \|^{\frac{1}{2}} \gtrsim \sqrt{n}.\nonumber
\end{eqnarray}

Combining equations \eqref{eq:epsilon} and \eqref{eq:min} we get
\beq \frac{\|F\|_{\min}}{\|F\|_{\epsilon}} \gtrsim \frac{\sqrt{n}}{\sqrt{\log n}}\eeq.
\end{proof}
%\smallskip

For the specific $F$ which was appeared in the above proof, we have following upper bound.
\begin{prop}\label{prop:upperbound}
The element $F=(V\otimes W)(a)\in \ell_1^n(\ell_{\infty}^n)\otimes B(H_n)$ defined in \eqref{e:F} %the proof Theorem \ref{thm:largeviolation}
verifies
\beq \|F\|_{\ell_1^n(\ell_{\infty}^n)\otimes_{\min} B(H_n)}\lesssim  \sqrt{n \log n}.\eeq
\end{prop}

\begin{proof}
Let us consider the map $W:\ell_2^{[\delta n]} \to B(H_n)$ defined in \eqref{e:W}.
By the result of Pisier (\cite{Pisier2003}), we have
\begin{equation}\begin{split}
\|W: OH_{\delta n} \to B(H_n)\|_{cb}^2 &= \left\|\sum_{k=1}^{[\delta n]} |1\rangle\langle k| \otimes |1\rangle\langle k|\right\|_{B(H_n)\otimes_{\min} B(H_n)}\\
& \leq\left\|\sum_{k=1}^{[\delta n]} |1\rangle\langle k|k\rangle\langle 1|\right\|^{\frac{1}{2}} \; \left\|\sum_{k=1}^{[\delta n]} |1\rangle\langle k|1\rangle\langle k|\right\|^{\frac{1}{2}}  \\
& \lesssim \sqrt{n}.
\end{split}
\end{equation}
Hence $\|W: OH_{[\delta n]} \to B(H_n)\|_{cb} \lesssim n^{\frac{1}{4}}.$ On the other hand, we have following fact be showed in \cite{JP2011}:
\beq \|V: OH_{[\delta n]} \to \ell_1^n(\ell_\8^n)\|_{cb}\lesssim n^{\frac{1}{4}} \sqrt{\log n}.\eeq
Therefore, we obtain:
\begin{equation}\begin{split}
\|F\|_{\ell_1^n(\ell_{\infty}^n)\otimes_{\min} B(H_n)} &=\left\|(V\otimes W)(\sum_k e_k \otimes e_k)\right\|_{\ell_1^n(\ell_{\infty}^n)\otimes_{\min} B(H_n)} \\
& \lesssim \sqrt{n\log n} \left\|\sum_k e_k \otimes e_k\right\|_{OH_{[\delta n]}\otimes_{\min} OH_{[\delta n]}}\\
& =\sqrt{n\log n} \left\|\sum_k e_k\otimes e_k\right\|_{\ell_2^{[\delta n]} \otimes_{\epsilon} \ell_2^{[\delta n]}} = \sqrt{n\log n}.
\end{split}\end{equation}
\end{proof}

Since our work is an adaptation of \cite{JP2011} in steering scenario,
 it is natural for us to provide the following two results:
\vspace{2mm}

{\bf Explicit form of the violation.}
Let $\epsilon_{x,a}^k, x,a,k=1,\ldots,n$ be independent Bernoulli sequences and let $K$ be a positive constant. Then we define:

\begin{enumerate}[{\rm i)}]
\item Steering functional $F_x^a \in B(H_{n+1})$:
        \begin{equation}\label{eq:steering} F_x^a =
 \left \{ \begin{split}
 &\frac{1}{n}\sum_{k=2}^{n+1} \epsilon_{x,a}^{k-1} |1\rangle\langle k| & \quad\quad x,a=1,\ldots,n,\\
 &0 & \quad \quad a=n+1.
 \end{split} \right.
 \end{equation}

\item POVMs measurements (\cite{JP2011}) $\{E_x^a\}_{x,a=1}^{n,n+1}$ in $B(H_{n+1})$ as
       \beq E_x^a =
 \left \{ \begin{split}
 &\frac{1}{n K}
 \begin{pmatrix}
 1 & \epsilon_{x,a}^1 & \cdots    & \epsilon_{x,a}^n \\
 \epsilon_{x,a}^1 & 1 & \cdots    & \epsilon_{x,a}^1\epsilon_{x,a}^n\\
 \vdots & \vdots &  & \vdots\\
 \epsilon_{x,a}^n & \epsilon_{x,a}^n\epsilon_{x,a}^1 & \cdots   & 1
 \end{pmatrix}.& \quad\quad a=1,\ldots,n,\\
 &1-\sum_{a=1}^n E_x^a & \quad \quad a=n+1
 \end{split} \right.
 \eeq
 for $x=1,\ldots,n.$
\item States: If $(\alpha_i)_{i=1}^{n+1}$ is a decreasing and positive sequence then set
\beq |\varphi_{\alpha} \rangle = \sum_{i=1}^{n+1} \alpha_i |ii\rangle.\eeq
\end{enumerate}
For this steering functional, the LHS and quantum bound can be calculated easily through the argument of \cite{JP2011}:

\begin{equation}\label{eq:5}
B_C(F)=\sup\left\{ \left|\sum_{x=1}^n\sum_{a=1}^{n+1} Tr (F_x^a \sigma_x^a)\right|: \; \sigma_x^a \in \mathcal{L}\right\} \leq C \sqrt{\log n},
\end{equation}
and
\begin{equation}\label{eq:6}
B_Q(F)\geq \sum_{x=1}^n \sum_{a=1}^{n+1} \Tr (F_x^a \Tr_A ( E_x^a \otimes \un | \varphi_\alpha\rangle \langle \varphi_\alpha|)) = \frac{1}{K} \alpha_1 \sum_{k=2}^{n+1} \alpha_k.
\end{equation}

For \eqref{eq:5}, define two maps $\widetilde{V}: \ell_2^{n+1} \to \ell_1^n (\ell_\8^{n+1})$ and $\widetilde{W}: \ell_2^{n+1} \to B(H_{n+1})$ as follows:
 \beq \widetilde{V}(e_k)=
 \left \{ \begin{split}
 &0 & k=1,\\
 &\frac{1}{n}\sum_{x=1}^n\sum_{a=1}^{n} \epsilon_{x,a}^{k-1}e_x \otimes e_a' & \quad\quad k=2,\ldots,n+1,
 \end{split} \right.
 \eeq
and
 \beq \widetilde{W}(e_k)=
 \left \{ \begin{split}
 &0 & k=1,\\
 &|1\rangle\langle k| & \quad\quad k=2,\ldots,n+1.\\
 \end{split} \right.
 \eeq
By \cite[Lemma 3.5]{JP2011}, we get
\begin{equation}\begin{split}
\mathbb{E}\|\widetilde{V}: \ell_2^{n+1} \to \ell_1^n(\ell_{\8}^{n+1})\| \leq C \sqrt{\log n}.
\end{split}
\end{equation}
Then by Chebyshev's inequality, with ``high probability" we can choose $\{\epsilon_{x,a}^k\}$ such that: $\|\widetilde{V}\| \leq C\sqrt{\log n}.$ Moreover, it is easy to see the map $\widetilde{W}: \ell_2^{n+1} \to B(H_{n+1})$ is a contraction, i.e. $\|\widetilde{W}\| \leq 1.$\\
Hence, by Proposition \ref{prop:minandepsilon}, we have
\begin{equation}\begin{split}
B_C(F) &\leq \left\|\sum_{x=1}^n \sum_{a=1}^{n+1} e_x \otimes e_a' \otimes F_x^a\right\|_{\ell_1^n(\ell_\8^{n+1}) \otimes_{\epsilon} B(H_{n+1})}  \\
&= \left\|\widetilde{V} \otimes \widetilde{W} (\sum_{k=1}^{n+1} e_k \otimes e_k)\right\|_{\ell_1^n(\ell_\8^{n+1}) \otimes_{\epsilon} B(H_{n+1})}  \leq C \sqrt{\log n}.
\end{split}
\end{equation}

For given $\alpha\in (0,1),$ let us consider $|\varphi_\alpha\rangle = \alpha|11\rangle + \sum_{i=2}^{n+1}\sqrt{\frac{1-\alpha^2}{n}} |ii\rangle$. It follows from the above result that $B_Q(F)\geq \frac{1}{K}\alpha\sqrt{1-\alpha^2} \sqrt{n} \gtrsim \sqrt{n}$. So, we have constructed explicitly a steering functional $F$ such that $LV(F)\gtrsim \frac{\sqrt{n}}{\sqrt{\log n}}$. %And
Let us mention that this unbounded largest violation is obtained by a non maximally entangled state.

\begin{remark}
This construction is explicit but also probabilistic. It does not guarantee that a given functional will yield unbounded largest violation. It happens with high probability.
\end{remark}

Another natural result is:
\vspace{2mm}

{\bf Larger steering violation by non maximally entangled state.}
Let $\rho= | \psi_d \rangle \langle \psi_d |$ be the $d$-dimensional maximally entangled state, where $ |\psi_d \rangle = \frac{1}{\sqrt{d}} \sum_{i=1}^d |ii\rangle$. In \cite[Theorem 5.1]{JP2011} the authors provide an example of a Bell functional which gives Bell violations of order $\frac{\sqrt{n}}{\log n},$ but only bounded violations can be obtained by any maximally entangled state. It is not surprising that we have similar conclusion in steering scenario. Following notion is crucial \cite{JP2011,JP1995}: Given two operator spaces $E$ and $F$, for any $a \in E \otimes F$ we define its $\psi-\min$ norm:
\beq \|a\|_{\psi-\min} = \sup | \langle \psi_d | (u\otimes v) (a)| \psi_d \rangle|,\eeq
where the supremum runs over all $d$ and all completely contractions $u: X\to B(H_d)$ and $v: Y\to B(H_d)$.

The next lemma follows directly from Proposition \ref{prop:minandepsilon} (or see \cite{JPPVW2010}).%, i.e. for every POVMs $\{ E_x^a\}_{x,a=1}^{n,m}$ on $B(H_d),$ the application $u: \ell_1^n(\ell_\8^m) \to B(H_d)$ given by $u(e_x\otimes e_a) =E_x^a$ for every $x, a,$ is a completely contraction, we have following lemma:
\begin{lem}\label{lem:modifiedminnorm}
Given an element $F=\sum_{x=1}^n\sum_{a=1}^{m} (e_x \otimes e_a') \otimes F_x^a \in \ell_1^n(\ell_\8^m) \otimes B(H_d),$ we have:
\beq \sup_{\Theta_{\max}\in\mathcal{Q}_{| \psi_d\rangle}} |\langle F, \Theta_{\max}\rangle| \lesssim \|F\|_{\psi-\min},\eeq
where $%\Theta_{\max}=(\sigma_x^a)\in
\mathcal{Q}_{| \psi_d\rangle} = \{ (\sigma_x^a)=(Tr_A (E_x^a \otimes \un_B) |\psi_d\rangle\langle\psi_d|): \mbox{{\em $\{ E_x^a\}_{x,a=1}^{n,m}$ is a POVMs}}\}$, i.e, $\mathcal{Q}_{| \psi_d\rangle}$ is the set of all assemblages which are constructed by the $d$-dimensional maximally entangled state. %and $\sup_{\mathcal{Q}_{| \psi_d\rangle}} |\langle F, \Theta_{\max}\rangle|$ is the quantum bound obtained from the $d$-dimensional maximally entangled state.
\end{lem}

Recall the following fact in \cite[Theorem 5.1]{JP2011}:
\begin{fact}
There are linear maps $S: R_n \cap C_n \to \ell_1^k (\ell_\8^{Dn})$ and $S^*:  \ell_1^k (\ell_\8^{Dn}) \to R_n \cap C_n$ such that
\beq S^* S =\mathrm{id}_{\ell_2^n}, \quad  \|S^*\|_{cb}\leq C, \quad\mbox{and} \quad \|S\|_{cb}\leq C \sqrt{\log n},\eeq where $k\leq 2^{D^2 n^2}.$
\end{fact}

Similarly to the proof of Theorem \ref{thm:largeviolation}, we can %show the
construct an element: $F=(S \otimes W)(\sum_k e_k \otimes e_k) \in \ell_1^k(\ell_\8^{Dn}) \otimes B(H_n)$ satisfying:
\begin{equation}\label{eq:3}
\|F\|_{\epsilon} \lesssim \sqrt{\log n}  \quad \mbox{and} \quad \|F\|_{\min} \gtrsim \sqrt{n}.
\end{equation}
Moreover,
\begin{equation}\label{eq:4}\begin{split}
\|F\|_{\psi-\min}&=\left\|S\otimes W (\sum_k e_k \otimes e_k)\right\|_{\psi-\min}\leq \|S\|_{cb} \|W\|_{cb}\left\|\sum_k e_k \otimes e_k\right\|_{\psi-\min}\\
&\leq C \sqrt{\log n} \left\|\sum_k e_k \otimes e_k\right\|_\epsilon \lesssim \sqrt{\log n}.
\end{split}\end{equation}

Now, we can represent $F$ as $F= \sum_{x,a} e_x \otimes e_a' \otimes F_x^a,$ %then
and define $\widetilde{F}_x^a\in \mathbb{M}_{n+1},$ such that the left-top $n\times n$ corner of $\widetilde{F}_x^a$ is $F_x^a$ and other coefficients of $\widetilde{F}_x^a$ are zeros.
%If we combine with
Equations \eqref{eq:3}, \eqref{eq:4} and Lemma \ref{lem:modifiedminnorm} %. We can conclude
lead to the conclusion that there exists a steering functional \beq \widetilde{F}=\{\widetilde{F}_x^a \in B(H_{n+1}): x=1,\ldots, 2^{n^2},\; a=1,\ldots,n+1\},\eeq
such that:
 \begin{enumerate}[{\rm i)}]
\item $LV(\widetilde{F}) \gtrsim \frac{\sqrt{n}}{\sqrt{\log n}},$
\item $\sup_{\Theta_{\max}\in\mathcal{Q}_{| \psi_n\rangle}} |\langle \widetilde{F}, \Theta_{\max}\rangle| \lesssim  \sqrt{\log n}.$
\end{enumerate}

This steering functional $\widetilde{F}$ has an unbounded largest violation of order $\frac{\sqrt{n}}{\sqrt{\log n}},$ but the unbounded largest violation can never been obtained by the maximally entangled state.

\begin{remark}\label{rmk:Jungework}
Since in steering scenario we consider the assemblages instead of joint probabilities, %so
the algebraic tensor product we consider is $\ell_1^n(\ell_\8^n) \otimes B(H_n).$ But in some sense, $B(H_n)$ is more friendly compare to $\ell_1^n(\ell_\8^n)$. That is the reason why we can easily apply Junge-Palazuelos's approach \cite{JP2011} to steering scenario.
\end{remark}

\subsection{Steering violation by partially entangled states including PPT states}
Here we will consider the role of the partially entangled state in violation of steering inequalities.
Recall that the $n$-dimensional pure partially entangled state is of the form
$|\psi_\alpha \rangle = \sum_{i=1}^n \alpha_i |ii\rangle,$ where $\alpha_i > 0$ and $\sum_i \alpha_i^2 =1.$
We have shown in Subsection 2.4 that it is possible to obtain
%For the steering functional $F_x^a$ \tblue{given in \eqref{eq:steering}}, through this partially entangled state, we can obtain
an unbounded largest violation of the order $\frac{\sqrt{n}}{\sqrt{\log{n}}}$
by means of a partially entangled state (for the steering functional given in \eqref{eq:steering}).

For any state $\rho$ and a steering functional $F=(F_x^a)$ let us consider
%For any given steering functional $F_x^a, x, a= 1,\ldots, n,$
the quantum bound $B_{Q_{\rho}}(F)$ obtained by means of a partially entangled state $|\psi_\alpha\rangle$, i.e.
\begin{equation*}
B_{Q_\rho}(F)=\sup\Big\{\Big|\sum_{x,a}\Tr(F_x^a\Tr_A((E_x^a\otimes\un)\rho)\Big|:\,\mbox{$(E_x^a)_a$ are POVMs for any $x$}\Big\}.
\end{equation*}
Given a partially entangled state $|\psi_\alpha\rangle$ for any steering functional $F=(F_x^a)_{x,a=1,\ldots,n}$ we have
\begin{equation}\begin{split}
B_{Q_{|\psi_\alpha\rangle}}(F) & = \sup \Big\{ \Big|\sum_{x,a=1}^n \Tr (F_x^a \Tr_A ((E_x^a \otimes \un) |\psi_\alpha\rangle\langle \psi_\alpha|))\Big|: E_x^a \; \text{be POVMs}\Big\}\\
& = \sup \Big\{\Big| \sum_{x,a} \sum_{i,j} \alpha_i \alpha_j \langle j| E_x^a |i\rangle \langle j|F_x^a| i \rangle\Big|: E_x^a \Big\}\\
& \leq \sum_{i,j} \alpha_i \alpha_j \sup \Big\{ |\sum_{x,a} \Tr (E_x^a |i\rangle\langle j|) \Tr (F_x^a |i\rangle\langle j|)|: E_x^a \Big\}\\
& \leq \sum_{i,j} \alpha_i \alpha_j \sup\{ \|\sum_{x,a} E_x^a \otimes F_x^a \|_{\mathbb{M}_n \otimes_\epsilon \mathbb{M}_n} : E_x^a \}\\
& \leq \sum_{i,j} \alpha_i \alpha_j \|\sum_{x,a} e_x \otimes e_a \otimes F_x^a \|_{\ell_1^n(\ell_\8^n) \otimes_\epsilon \mathbb{M}_n}.
\end{split}
\end{equation}
The second inequality %, we have used
follows from the fact that: $\Tr(\cdot \; |i\rangle\langle j|) \in \mathbb{B}_{S_1^n}$,
where $\mathbb{B}_X$ denotes the unit ball in a norm space $X$.
Now, for any $\lambda\in[0,1]$ let us %we
consider the following density matrix
$$\rho_\lambda= (1-\lambda) \frac{\un}{n^2} + \lambda |\psi_\alpha\rangle\langle \psi_\alpha|.$$%, 0\leq \lambda \leq 1.$
By the preceding discussion, the quantum bound %obtained from $\rho_\lambda$, denote by
$B_{Q_{\rho_\lambda}}(F),$ is bounded by
\beq (1-\lambda+ \lambda \sum_{i,j} \alpha_i \alpha_j)\|\sum_{x,a} e_x \otimes e_a \otimes F_x^a \|_{\ell_1^n(\ell_\8^n) \otimes_\epsilon \mathbb{M}_n}.\eeq
Thus,
\begin{equation}\label{eq:8a}
B_{Q_{\rho_\lambda}}(F) \lesssim (1-\lambda+ \lambda \sum_{i,j} \alpha_i \alpha_j) B_C(F).
\end{equation}
In \cite{Pusey2013}, the author presented a stronger version of Peres conjecture: ``PPT states can not violate steering inequalites, i.e, the assemblages obtained by measuring them always have LHS models." The conjecture has been disproved in
\cite{MGHG2014}. However, one can still ask, whether PPT states can exhibit unbounded violation. In the sequel %following
we will consider two classes
of PPT states, and show that they allow only bounded steering violation.

Firstly, let us consider %considering
PPT states among states $\rho_\lambda$ for a given partially entangled state $|\psi_\alpha\rangle$.
The partial transpose of $\rho_\lambda$ has the following eigenvalues
 \begin{equation}
 \left \{ \begin{split}
 &\frac{1-\lambda}{n^2} \pm \lambda \alpha_i \alpha_j & \quad\quad i \neq j,\\
 &\frac{1-\lambda}{n^2} +\lambda \alpha_i^2 & i=1,\ldots, n.
 \end{split} \right.
 \end{equation}
Thus $\rho_\lambda$ is a PPT state if and only if $\lambda \leq \min\{ \frac{1}{1+ n^2 \alpha_i \alpha_j}: i \neq j\}.$ From the equation \eqref{eq:8a}, we have
\begin{equation}\begin{split}
B_{Q_{\rho_\lambda}}(F) & \lesssim \left(1-\lambda+ \lambda \sum_{i,j} \alpha_i \alpha_j\right) B_C(F)\\
& \leq \left(1+ \min\left\{ \frac{n-1}{1+ n^2 \alpha_i\alpha_j}: i\neq j\right\}\right) B_C(F)\\
& \leq \left(1+ \frac{n-1}{1+ n^2 \alpha \sqrt{\frac{1-\alpha^2}{n-1}}} \right) B_C(F),
\end{split}\end{equation}
where $\alpha=\max_i \{\alpha_i\}.$ It follows from $\sum_i \alpha_i^2=1$ that $\alpha \geq \frac{1}{\sqrt{n}}.$ Thus $B_{Q_\lambda}(F) \lesssim B_C(F).$
Now we can make following conclusion:
For PPT states $\rho_\lambda= (1-\lambda) \frac{\un}{n^2} + \lambda |\psi_\alpha\rangle\langle \psi_\alpha|, \lambda \leq \min\{ \frac{1}{1+ n^2 \alpha_i \alpha_j}: i \neq j\},$ and any given steering functional $F$, the quantum bound obtained by $\rho_\lambda$ is bounded by the LHS bound up to an universal constant.

In \cite{CK2006}, D. Chru\'{s}ci\'{n}ski and A. Kossakowski introduced another class of PPT states. It is invariant under the maximal commutative subgroup of $\mathbb{U}(n),$ and it includes the previous isotropic states $\rho_\lambda$. Briefly speaking, they considered following two classes of PPT states:
 \begin{enumerate}[{\rm i)}]
\item Isotropic-like states, $\rho= \sum_{i,j=1}^n a_{ij}|ii\rangle\langle jj| + \sum_{i\neq j=1}^n c_{ij} |ij\rangle\langle ij|, \;\; 	(a_{ij})_{i,j}\geq 0, c_{ij} \geq 0, c_{ij}c_{ji}-|a_{ij}|^2 \geq 0, \sum_{i=1}^n a_{ii} + \sum_{i\neq j=1}^n c_{ij}=1 ;$
\item Werner-like state, $\widetilde{\rho}= \sum_{i,j=1}^n b_{ij} |ij\rangle\langle ji| + \sum_{i\neq j=1}^n c_{ij}|ij \rangle\langle ij|, \;\; (b_{ij})_{i,j}\geq 0, c_{ij}\geq 0, c_{ij} c_{ji}-|b_{ij}|^2 \geq 0, \sum_{i=1}^n b_{ii} + \sum_{i\neq j=1}^n c_{ij}=1.$
\end{enumerate}

Now we will prove that for any steering functional the largest violation obtained by means of states from the above two classes of PPT states is bounded. To see this, we will use the following proposition, which is an analogue of \cite[Theorem 2.1]{P2012}.

\begin{prop}
Given an $n$-dimensional bipartite state $\rho \in S_n^1 \otimes S_n^1$ for any steering functional $(F_x^a)_{x,a=1,\ldots,n},$ we have
\beq  B_{Q_\rho}(F) \leq \|\rho\|_{S_n^1\otimes_{\pi} S_n^1}  B_C(F).\eeq
%where $B_{\rho}(F) = \sup \Big\{ \Big| \sum_{x,a=1}^n \Tr (F_x^a \Tr_A(E_x^a \otimes \un \rho) )\Big|: E_x^a \; \text{be POVMs}\Big\}.$
\end{prop}

\begin{proof}
The proof is more or less the same as the proof in \cite{P2012}. By duality and Proposition \ref{prop:minandepsilon}, we have:
\begin{equation}
\begin{split}
\Big| \sum_{x,a} \Tr (F_x^a \Tr_A(E_x^a \otimes \un \rho) )\Big| & = \Big| \sum_{x,a} \Tr (E_x^a \otimes F_x^a \rho) \Big|\\
& \leq \Big\|\sum_{x,a} E_x^a \otimes F_x^a\Big\|_{\mathbb{M}_n \otimes_{\epsilon} \mathbb{M}_n} \|\rho\|_{S_n^1 \otimes_{\pi} S_n^1}\\
& \leq \Big\|\sum_{x,a} e_x \otimes e_a' \otimes F_x^a \Big\|_{\ell_1^n(\ell_\8^n) \otimes_{\epsilon} \mathbb{M}_n} \|\rho\|_{S_n^1 \otimes_{\pi} S_n^1}\\
& \leq B_C(F) \|\rho\|_{S_n^1 \otimes_{\pi} S_n^1}.
\end{split}
\end{equation}
 \end{proof}

Now it remains to calculate the projective norm of Isotropic-like state and Werner-like state. It can be proved that the norms of both states are bounded by 2. For instance, for any Isotropic-like state $\rho,$
\begin{equation}
\begin{split}
\|\rho\|_{\pi} & \leq \sum_{i,j=1}^n |a_{ij}| + \sum_{i\neq j=1}^n c_{ij} \leq 1+ \sum_{i\neq j} |a_{ij}|\\
& \leq 1+ \sum_{i\neq j} \sqrt{c_{ij} c_{ji}} \leq 1+ \left( \sum_{i\neq j} c_{ij}\right)^2 \leq 2.
\end{split}
\end{equation}

\begin{remark}
The PPT  states described in \cite{CK2006} cover many PPT entangled states known in literature, however it does not describe
bound entangled states constructed via unextendible product bases (UPB) \cite{Bennett1999}. Unfortunately, up to now we can't estimate the projective norm of PPT states which are constructed by UPB.
\end{remark}

%\vspace{5mm}%

%%%%%%%%%%%%%%%%%%%%%%%%%
%%%%%%%%%%%%%%%%%%%%%%%%%

\section{Dichotomic case}

In \cite{PWPVM2008} the authors considered the dichotomic setting for Bell scenario (see also \cite{Pisier2012}). It is more or less a reformulation of the standard setting for the Bell scenario with two outcomes. It turns out that in the steering scenario it is no longer the case: standard and dichotomic settings are not equivalent. The details of this phenomenon are discussed in \cite{steering_short}. Here we describe its particular exemplification by using the operator space technique.

In the dichotomic setting for steering scenario we assume that the measurement for Alice has only two outcomes $\pm 1.$ Alice prepares two correlated particles sharing with a quantum state $\rho\in B(H_A \otimes H_B)$ and send one of them to Bob. Alice wants to convince Bob that $\rho$ is an entangled state by doing dichotomic measurement $-\un \leq E_x \leq \un, x=1,\ldots, n.$ After Alice's measurement has been done, Bob obtains the conditional states
\begin{equation}\label{eq:7}\sigma_x= \Tr_A (\rho (E_x \otimes \un)).\end{equation}
If the nature is described by a LHS model, then
\begin{equation}\label{eq:8}
\sigma_x= \sum_\lambda p(\lambda) E_x(\lambda) \sigma_\lambda,\end{equation}
where $p(\lambda)$ is a probability distribution function, $E_x(\lambda)= \pm 1$ is the deterministic outcome obtained by Alice if she does the measurement $E_x,$ and $\sigma_\lambda$ is a density matrix of $B(H_B).$

For $n=\dim(H_B)$ we define a steering functional as a set of $n\times n$ matrices $F_x, x=1,\ldots, n.$
Analogously to the standard case, we can define the LHS bound of $F$ as:
\beq B_C (F)=\sup \left\{\left|\sum_x\Tr(F_x \sigma_x)\right|: \sigma_x \;\text{satisfy} \;\eqref{eq:7} \right\},\eeq
and the quantum bound of $F$ as
\beq B_Q (F)=\sup \left\{\left|\sum_x \Tr(F_x \sigma_x)\right|: \sigma_x \;\text{satisfy} \; \eqref{eq:8} \right\}.\eeq
%Given $F=\sum_{x=1}^n e_x \otimes F_x \in \ell_1^n \otimes \mathbb{M}_n,$
One can also apply the argument of Proposition \ref{prop:minandepsilon} to obtain
$$B_C(F)\simeq \|F\|_{\ell_1^n \otimes_{\epsilon} \mathbb{M}_n}\quad\mbox{and}\quad B_Q(F) \simeq \|F\|_{\ell_1^n \otimes_{\min} \mathbb{M}_n}.$$

\begin{remark}
In \cite{PWPVM2008}, the authors considered the Bell scenario in dichotomic setting: for a Bell inequality $M=\sum_{x,y=1}^n M_{x,y} e_x \otimes e_y \in \ell_1^n \otimes \ell_1^n,$ the classical bound $B_C(M)$ is equivalent to the norm $\|M\|_{\ell_1^n \otimes_\epsilon \ell_1^n}$ and the quantum bound $B_Q(M) \simeq \|M\|_{\ell_1^n \otimes_{\min} \ell_1^n.}$ By Grothendieck's theorem \cite{Pisier2012}, we have $\ell_1^n \otimes_\epsilon \ell_1^n \simeq \ell_1^n \otimes_{\min} \ell_1^n.$ So for any bipartite dichotomic Bell functional $M$, there always exists an universal constant $K$ (not depend on the dimension), such that $B_Q(M) \leq K B_C(M).$ They also proved that for tripartite case, this universal constant does not exist!
\end{remark}

The situation in steering scenario differs from the features of the Bell scenario described in the above remark. Since $\ell_1^n \otimes_\epsilon \mathbb{M}_n \ncong \ell_1^n \otimes_{\min} \mathbb{M}_n$, one should expect that there is a room for a steering functional with unbounded largest violation. Such a functional has been provided in %Ref.
\cite{steering_short}. Here we restate this result as the
following theorem, and provide a proof referring  to operator space formalism.

\begin{thm}\label{thm:unbounded}
For every $n\in \mathbb{N},$ there exists a steering functional $(F_x)_{x=1,\ldots,n} \in \ell_1^n\otimes\mathbb{M}_n$ with unbounded largest violation of order $\sqrt{\log{n}}$.
\end{thm}
\begin{proof}
It is enough to prove the following claim: For every $n\in \mathbb{N},$ there exists an element $F=\sum_{x=1}^{n} e_x \otimes F_x \in \ell_1^n \otimes \mathbb{M}_n,$ such that
\beq \frac{\|F\|_{\min}}{\|F\|_{\epsilon}} \gtrsim \sqrt{\log{n}}.\eeq
Since we know that $\ell_1^n \otimes_\epsilon \mathbb{M}_n = B(\ell_\8^n, \mathbb{M}_n)$ isometrically and $\ell_1^n \otimes_{\min} \mathbb{M}_n= CB(\ell_\8^n, \mathbb{M}_n)$ completely isometrically, it is enough to prove that there exists a map $\phi$ from $\ell_\8^n$ to $\mathbb{M}_n,$ such that $\|\phi\|\lesssim 1$ and $\|\phi\|_{cb}\gtrsim \sqrt{\log{n}}.$ Here we will use a fact described in \cite{L1976}: there is a map $\varphi: \ell_\8^n \to \mathbb{M}_{2^n},$ such that $\varphi$ is bounded but not completely bounded. For the  reader's convenience, we just rewrite their proof. Let $\sigma_i, i=x,y,z$ be Pauli matrices. Let
\begin{equation}
\begin{split}
A_1 &= \sigma_x \otimes \overbrace{\un\otimes \ldots \otimes \un}^{n-1}, \\
A_2 &= \sigma_z \otimes \sigma_x \otimes\overbrace{ \un\otimes \ldots \otimes \un}^{n-2},\\
& \vdots\\
A_k &= \overbrace{\sigma_z \otimes \ldots \otimes \sigma_z}^{k-1} \otimes \sigma_x \otimes \overbrace{\un \otimes \ldots \otimes \un}^{n-k},\; 3\leq k \leq n.
\end{split}
\end{equation}
It is easy to check that these $A_i \in \mathbb{M}_{2^n}, i=1,\ldots,n$ satisfies $A_i = A_i ^*, A_i A_j+A_j A_i = 2\delta_{ij} \un_{2^n}.$ For any $A= \sum_i a_i A_i \in \mathbb{M}_{2^n},$ since $A^*A+AA^*=2\sum_i|a_i|^2 \un_{2^n},$ then $\|A\|\leq \sqrt{2}\sqrt{\sum_i |a_i|^2}.$ Now we define the map $\varphi: \ell_\8^n \to \mathbb{M}_{2^n}$ as $\varphi(e_i)=\frac{1}{\sqrt{n}} A_i, i=1,\ldots,n.$ Then
\begin{equation}
\left\|\varphi\left(\sum_i a_i e_i\right)\right\| = \left\| \frac{1}{\sqrt{n}} \sum_i a_i A_i\right\|\leq \sqrt{\frac{2}{n}} \sqrt{\sum_i |a_i|^2}\leq \sqrt{2} \sup_i |a_i|.
\end{equation}
Thus $\|\varphi\| \lesssim 1.$ On the other hand, we let $\theta= \sum_i A_i \otimes e_i \in \mathbb{M}_{2^n} \otimes \ell_\8^n.$ Note that $\|\theta\|=\sup_i \|A_i\|=1,$ and by using the fact that \cite{L1976}: there is a unit vector $z\in \mathbb{C}^{2^n} \otimes \mathbb{C}^{2^n}$ such that $(A_i \otimes A_i)(z)=z$ for any $i$. Then
\begin{equation}
\begin{split}
\|\varphi\|_{cb} & \geq \|\un_{\mathbb{M}^{2^n}} \otimes \varphi(\theta)\| = \left\|\frac{1}{\sqrt{n}}\sum_i A_i \otimes A_i\right\|\\
& \geq \frac{1}{\sqrt{n}} \left|\left\langle \sum_i (A_i\otimes A_i)(z), z \right\rangle\right|= \frac{1}{\sqrt{n}} \left|\sum_i \langle z,z \rangle\right|=\sqrt{n}.
\end{split}
\end{equation}

Since for very natural number $n\geq 2,$ there exists natural number $m,$ such that $n \geq 2^m,$
consider the diagram
\beq  \ell_\8^n \stackrel{\omega_1}{\longrightarrow} \ell_\8^m \stackrel{\varphi}{\longrightarrow}  \mathbb{M}_{2^m} \stackrel{\omega_2}{\longrightarrow} \mathbb{M}_n,\eeq
where $\omega_1$ projects $\ell_\8^n$ onto the first $m$ coordinates, and $\omega_2$ embeds $\mathbb{M}_{2^m}$ into the top $2^m\times 2^m$ corner of $\mathbb{M}_n.$ Set $\phi= \omega_2\circ \varphi \circ \omega_1: \ell_\8^n \to \mathbb{M}_n,$ then
\beq \|\phi\|= \|\varphi\|\lesssim 1; \;\; and \;\; \|\phi\|_{cb}= \|\varphi\|_{cb}\geq \sqrt{m}\geq \sqrt{\log{n}}.\eeq
Thus we can find a map $\phi: \ell_\8^n \to \mathbb{M}_n,$ such that $\frac{\|\phi\|_{cb}}{\|\phi\|} \gtrsim \sqrt{\log{n}}.$
If $F=\sum_{x=1}^n e_x \otimes \phi(e_x),$ then $F$ satisfies the statement of the theorem.
\end{proof}

\begin{remark}
This result also can be traced back to the work of Paulsen \cite{Paulsen1992} (or see \cite[Section 3.3]{Pisier2003}). For any norm space $E$, he defined a constant
\begin{equation}
\alpha (E) = \|id: \min(E) \mapsto \max(E)\|_{cb},
\end{equation}
where $\min(E)$ ($\max(E)$) is the minimal (maximal) admissible operator space structure of $E.$ It can be proved \cite{Paulsen1992} that this constant is equal to
$\sup \{\|T\|_{cb}: \|T\|\leq 1, T: \min(E) \mapsto B(H)\},$ $H$ is arbitrary. For $E= \ell_\infty^n,$  due to Loebl \cite{L1976} and Haagerup \cite{Haagerup1983}'s work, we have $\alpha(\ell_\infty^n) \geq \sqrt{\frac{n}{2}}$ \cite{Paulsen1992}.
\end{remark}

From this theorem, in the dichotomic case, the unbounded largest violation derives from some bounded but not completely bounded map. Now we will discuss for what kind of steering functional $F_x,$ we can always get bounded largest violation. For any positive steering functional $(F_x)_{x=1,\ldots,n},$ i.e., $F_x \geq 0$ for every $x,$ the LHS bound:
\begin{equation}\begin{split}
B_C (F) &= \sup \left\{\left|\sum_x \Tr\left(F_x \sum_\lambda p(\lambda) E_x(\lambda) \sigma_\lambda \right)\right| \right\}\\
&= \sup_\lambda  \left\|\sum_x E_x (\lambda) F_x\right\|_{\mathbb{M}_n} \leq \left\|\sum_x  F_x\right\|_{\mathbb{M}_n}.
\end{split}
\end{equation}
On the other hand,
\begin{equation}\begin{split}
B_C (F) &\approx \left\|\sum_x e_x \otimes F_x \right\|_{\ell_1^n \otimes_{\epsilon} \mathbb{M}_n} \\
&= \sup \left\{\left|\sum_x f(e_x) g(F_x) \right|: \;\; f\in \mathbb{B}_{\ell_\8^n}, g\in \mathbb{B}_{S_1^n} \right\} \geq \left\|\sum_x F_x\right\|_{\mathbb{M}_n}.
\end{split}
\end{equation}
Thus $B_C(F)\approx \left\|\sum_x F_x\right\|_{\mathbb{M}_n}.$

For the quantum bound,
$B_Q(F) \approx \left\|\sum_x e_x \otimes F_x \right\|_{\ell_1^n \otimes_{\min} \mathbb{M}_n}.$ It is known that \cite{JPPVW2010,Pisier2003}
\begin{equation}\begin{split}
\Big\|\sum_x e_x \otimes F_x \Big\|_{\ell_1^n \otimes_{\min} \mathbb{M}_n} &= \sup \Big\{ \Big\| \sum_x U_x \otimes F_x \Big \|_{\mathbb{M}_n \otimes_{\min} \mathbb{M}_n} : \;\; U_x \in \mathbb{M}_n, U_x U_x^*= U_x^* U_x = \un.\Big\}\\
&= \inf \Big\{ \Big\| \sum_x b_x b_x^* \Big\|^{\frac{1}{2}} \; \Big\| \sum_x c_x^* c_x \Big\|^{\frac{1}{2}}: \; F_x= b_x c_x.\Big\}
\end{split}
\end{equation}
If $F_x$ is a positive matrix for every $x=1,\ldots,n,$ then by the lemma 2 of \cite{JPPVW2010}, we know
\beq \|\sum_x e_x \otimes F_x \|_{\ell_1^n \otimes_{\min} \mathbb{M}_n} = \|\sum_x F_x \|_{\mathbb{M}_n}.\eeq

We end this section with following remark.
\begin{remark}
If $F_x \geq 0,$
then $\|\sum_x e_x \otimes F_x\|_{\ell_1^n \otimes_{\min} \mathbb{M}_n} \approx  \|\sum_x e_x \otimes F_x\|_{\ell_1^n \otimes_{\epsilon} \mathbb{M}_n}.$ In other words, the quantum bound of positive dichotomic steering functional is always bounded by its LHS bound.
\end{remark}

\section{Conclusion}
In this paper, we have used operator space approach to study violation of steering inequality. Before, the approach had been successfully used in Bell scenario \cite{JP2011,JPPVW2010,PWPVM2008}. In both cases, operator space was connected to the largest violation of corresponding functional. The main difference is the algebraic tensor product considered in each case. In Bell scenario, the algebraic tensor product is $\ell_1^n(\ell_\8^m) \otimes \ell_1^n(\ell_\8^m),$ while in steering scenario, it is $\ell_1^n(\ell_\8^m) \otimes B(H_d).$ Since our work is an extension of applying Junge-Palazuelos approach to steering scenario, we can easily construct a probabilistic steering functional with unbounded largest violation. And for this functional, non maximally entangled state will give larger violation. However, not all properties of steering functionals can be recovered from the Bell scenario. We have shown in \cite{steering_short}  a phenomenon characteristic only for the steering scenario, i.e, there is a bipartite dichotomic steering functional with unbounded largest violation. In this paper, we have studied this phenomenon in the framework of operator space theory.

\section*{Acknowledgements}
We would like to thank Prof. W.A. Majewski for valuable remarks and fruitful discussion. The work is supported by
Foundation for Polish Science TEAM project co-financed by the EU European Regional Development
Fund, Polish Ministry of Science and Higher Education Grant no. IdP2011 000361, ERC AdG grant QOLAPS and  EC grant RAQUEL and a NCBiR-CHIST-ERA Project QUASAR.
Part of this work was done in National Quantum Information Center of
Gda{\'n}sk. Part of this work was done when the authors attended the program “Mathematical
Challenges in Quantum Information” at the Isaac Newton Institute for Mathematical
Sciences, University of Cambridge.
%\bigskip

%%%%%%%%%%%%%%%%%%%%%%%%%%%
%%%%%%%%%%%%%%%%%%%%%%%%%%%

\end{document}